\documentclass[12pt]{article}
{}
\newcommand{\rank}{\textrm{rank}}
\usepackage{indentfirst, anysize, hyperref, graphicx, amssymb,amsmath, amssymb, amsmath}
\usepackage{amsthm}
\theoremstyle{plain}
\newtheorem{theorem}{Theorem}[section]
\newtheorem{lemma}[theorem]{Lemma}

\newtheorem{proposition}[theorem]{Proposition}

\theoremstyle{definition}
\newtheorem{definition}[theorem]{Definition}

\theoremstyle{remark}

\newcommand{\keywords}[1]{\par\addvspace\baselineskip
\noindent Keywords: \enspace\ignorespaces#1}

\def\squareforqed{\hbox{\rlap{$\sqcap$}$\sqcup$}}
\def\qed{\ifmmode\squareforqed\else{\unskip\nobreak\hfil
\penalty50\hskip1em\null\nobreak\hfil\squareforqed
\parfillskip=0pt\finalhyphendemerits=0\endgraf}\fi}

\newenvironment{proofof}[1]{\begin{trivlist}%
\item[]{\flushleft\em Proof of #1. }}
{\end{trivlist}}

\begin{document}
  \title{Communication Complexities of Symmetric XOR functions}
\author{Yaoyun Shi\footnote{Department of Electrical Engineering and Computer Science, University of Michigan
1301 Beal Avenue, Ann Arbor, MI 48109-2122, USA, Email:
shiyy@eecs.umich.edu}\ \ \ \ \ \ \ Zhiqiang Zhang\footnote{Institute
for Theoretical Computer Science, Tsinghua University, Beijing,
100084, P. R. China. Email: zhang@itcs.tsinghua.edu.cn. Partially
supported by the National Science Foundation of the United States
under Awards 0347078 and 0622033, and the National Natural Science
Foundation of China Grant 60553001, the National Basic Research
Program of China Grant 2007CB807900 and 2007CB807901. It was
partially conducted while the author was visiting University of
Michigan, Ann Arbor.}}

  \maketitle

  \abstract{We call $F:\{0, 1\}^n\times \{0, 1\}^n\to\{0, 1\}$
a symmetric XOR function if for a function $S:\{0, 1, ...,
n\}\to\{0, 1\}$, $F(x, y)=S(|x\oplus y|)$, for any $x, y\in\{0,
1\}^n$,  where $|x\oplus y|$ is the Hamming weight of the bit-wise
XOR of $x$ and $y$.
  We show that for any such function, (a) the deterministic communication complexity is always $\Theta(n)$ except for four
  simple functions that have a constant complexity, and (b) up to  a polylog factor, the error-bounded randomized and
quantum communication complexities are $\Theta(r_0+r_1)$, where
$r_0$ and $r_1$ are the minimum integers such that
   $r_0, r_1\leq n/2$ and $S(k)=S(k+2)$ for all $k\in[r_0, n-r_1)$.}


\section{Introduction}

The two-party interactive communication model was introduced by
Yao~\cite{yao79} in 1979, and has been widely studied since then for
its simplicity and its power in capturing many of the complexity
issues of communication. Let $X$ and $Y$ be two sets and $F:X\times
Y\to\{0, 1\}$ be a Boolean function. In this model, Alice has an
input $x\in X$, Bob has an input $y\in Y$, and they want to compute
$F(x, y)$ by exchanging messages. If the communication protocol is
deterministic, the least number of bits they need to exchange on the
worst-case input is the {\em deterministic complexity}, denoted by
$D(F)$. If they are allowed to share random bits, the least number
of bits they need to exchange in order to compute $F$ with at least
$2/3$ of success probability is the {\em randomized complexity} of
$F$, denoted by $R(F)$. Yao also initiated the study of quantum
communication complexity~\cite{yao93}, denoted by $Q(F)$, which is
the least number of quantum bits that Alice and Bob need to exchange in
order to compute $F$ with at least $2/3$ of success probability for
any input. In this paper, we allow a quantum protocol to start with an
unlimited amount of quantum entanglement. Evidently, we have
$Q(F)\leq R(F)\leq D(F)$.

A major research theme in communication complexity is to identify
the asymptotic behavior of those variants of complexities for
specific and often elementary functions. A closely related focus is
to identify functions on which the maximum gaps among those
complexities are achieved. Despite numerous studies, both types of
questions are often difficult to answer. For an overview of the
field, an interested reader is referred to \cite{communicationbook,
Brassardsurvey,Buhrmansurvey,Sherstovsurvey}. In this paper, we
focus on the communication complexity of a class of functions that
we call {\em symmetric XOR} functions, and our main results are
tight, or almost tight, characterizations of their deterministic,
randomized and quantum complexities.

To state our main results, let us define the necessary notation.
Throughout this paper, the length of the inputs to Alice and Bob is
denoted by $n$. The Hamming weight of $z\in\{0, 1\}^n$ is denoted by
$|z|$. The bit-wise XOR of $x, y\in\{0, 1\}^n$ is denoted by $x\oplus
y$. A Boolean function $f:\{0,1\}^n\rightarrow\{0,1\}$ is {\em
symmetric} if $f(x)$ depends only on $|x|$, for all $x$.

\begin{definition} A communication problem $F:\{0, 1\}^n\times\{0, 1\}^n\to\{0, 1\}$
is a {\em  XOR function}  if for a  Boolean function $f:\{0,
1\}^n\to\{0, 1\}$, $F(x, y)=f(x\oplus y)$, for all $x, y\in\{0,
1\}^n$. It is said to be {\em symmetric} whenever $f$ is symmetric.
A symmetric XOR function is {\em trivial} if the function or its
negation has $f(x)=0$, for all $x$, or $f(x)=|x|\mod 2$, for all
$x$.
\end{definition}

If $f\equiv 0$, evidently $D(F)=0$. If $f$ is the XOR function,
$D(F)=1$ since it suffices for Alice to send $b=|x|\mod 2$ and Bob
then calculates $b+|y|\mod 2 = F(x, y)$. For nontrivial symmetric
XOR functions, we have the following.

\begin{theorem} \label{deterministic lower bounds}
   For any nontrivial symmetric XOR function $F:\{0, 1\}^n\times\{0, 1\}^n\to\{0, 1\}$, $D(F)=\Theta(n)$.
\end{theorem}

To prove the above result, we make use of the following fact that relates
$D(F)$ to the rank of the matrix $M_F=[f(x, y)]_{x, y\in\{0, 1\}^n}$,
denoted by $\rank(M_F)$.

\begin{lemma}[\cite{log rank is lower bound}]\label{lm:logrank}
For any $F:\{0, 1\}^n\times\{0, 1\}^n\to\{0, 1\}$,
$D(F)=\Omega(\log\rank(M_F))$.
\end{lemma}

It turns out that for a XOR function $F$, $\rank(M_F)$ is precisely
the number of non-zero Fourier coefficients of $f$. Recall that the
Fourier coefficient $\tilde f(w)$, where $w\in\{0, 1\}^n$, is
defined as

\begin{equation}\label{eqn:fourier}
\tilde f(w)=\frac{1}{2^n}\ \sum_{x\in\{0, 1\}^n} (-1)^{x\cdot w} f(x).
\end{equation}

Our main technical contribution is the following lemma.

\begin{lemma}\label{lm:fourier}
For all sufficiently large $n$, and any symmetric function $f:\{0,
1\}^n\to\{0, 1\}$ other than the constant $0$ function, the parity
function and their negations, there exists $w\in\{0, 1\}^n$ such that
$\tilde f(w)\ne 0$ and $n/8 \le |w|\le 7n/8$.
\end{lemma}

By the symmetry of $f$, the above lemma implies that $\tilde f$ has
$2^{\Omega(n)}$ non-zero Fourier coefficients, thus
$\rank(M_F)=2^{\Omega(n)}$. Theorem~\ref{deterministic lower bounds}
then follows from Lemma~\ref{lm:fourier}. Another consequence is that $D(M_F)=\Theta(\log\rank(M_F))$
for all symmetric XOR functions, since both $D(F)$ and $\rank(M_F)$
is a constant when $F$ is trivial. That is, symmetric XOR functions
satisfy the Log-Rank Conjecture of Lov{'a}sz and Saks~\cite{log rank
conjecture}, which states that for all Boolean functions $F$,
$D(F)=\log^{\Omega(1)}\rank(M_F)$.

We now turn to our second main result, which is on the randomized
and the quantum complexities of symmetric XOR functions. The
following two parameters of $F$ are critical to the complexities.
\begin{definition}
Let $F:\{0, 1\}^n\times\{0, 1\}^n\to\{0, 1\}$ be a  symmetric XOR
function, and $F(x, y)=S(|x\oplus y|)$, where
$S:\{0,1,\cdots,n\}\rightarrow \{0,1\}$. Define $r_0=r_0(F)$ and
$r_1=r_1(F)$ to be the minimum numbers $r_0'$ and $r_1'$,
respectively, such that
   $r'_0, r'_1\leq n/2$, and $S(k)=S(k+2)$, for any $k\in [r'_0, n-r'_1)$.
Define $r=r(F)=\max \{r_0, r_1\}$.
\end{definition}

\begin{theorem} \label{quantum and randomized complexity}
  For any symmetric XOR function $F:\{0, 1\}^n\times \{0, 1\}^n\to\{0, 1\}$,
  $Q(F)=\Omega(r)$, and $R(F)=O(r\log^2 r \log\log r)$.
\end{theorem}

A corollary of the above theorem is the confirmation of the
so-called Log-Equivalence Conjecture \cite{yaoyun block composed
functions}, when restricted to symmetric XOR functions. The
Log-Equivalence Conjecture states that quantum and randomized
communication complexities of any Boolean functions are polynomially
related.

Before we give the details of our proofs, we relate our results to
some other closely related works. That we focus on symmetric XOR
functions was inspired by Razborov's work~\cite{razborov symmetric
predicates} on what he called ``symmetric predicates'' and
subsequent works. A function $F:\{0, 1\}^n\times \{0, 1\}^n\to\{0,
1\}$ is a symmetric predicate if $F(x, y)=S(|x\wedge y|)$, where
$S:\{0, 1, ..., n\}\to\{0, 1\}$ and $x\wedge y\in\{0, 1\}^n$ is the
bit-wise AND of $x$ and $y$. Let $\ell_0$ and $\ell_1$ be the
minimal integers such that $\ell_0, \ell_1\leq n/2$ and that $S$ is
constant in $[\ell_0, n-\ell_1)$. Razborov showed that
$Q(F)=\Theta^*(\sqrt{n\ell_0}+\ell_1)$. Our quantum lower bound is a
technical consequence of Razborov's lower bound. Our classical upper
bound follows the same strategy of Huang et
al.~\cite{hammingdistance} in constructing a $O(d\log d)$-bits
randomized protocol to decide if $|x\oplus y|>d$.

We prove Theorem~\ref{deterministic lower bounds} by Fourier
analysis of Boolean functions, which is a powerful tool for the
study of Boolean functions complexity. The course
notes~\cite{odonnell} provide an excellent survey on the subject.
The closest result to Lemma~\ref{lm:fourier} that we are aware of is
by Lipton et al.~\cite{LiptonR} on a quantity $\Delta(n)$, which is
the minimum integer $n'$ such that any symmetric $f$ other than the
parity functions and the constant $0$ or $1$ functions
has a non-zero Fourier coefficient $\tilde f(w)$
with $1\le |w|\le n'$. They showed that
$\Delta(n)=\Omega(n/\log n)$, which has applications in
computational learning theory. Their method, however, does not seem
to be applicable for our question.

Finally, we note that class of XOR functions is a subset of three
classes of functions studied previously. (1) Shi and
Zhu~\cite{yaoyun block composed functions} studied what they called
{\em block-composed functions}, i.e. functions $F: \{0, 1\}^{k
n}\times \{0, 1\}^{kn}\to\{0, 1\}$ that can be represented as
$F(x_1, x_2, ..., x_n, y_1, y_2, ..., y_n)=f(g(x_1, y_1), g(x_2,
y_2), ..., g(x_n, y_n))$, for all $x_i, y_i\in\{0, 1\}^k$, and some
functions $f:\{0, 1\}^n\to\{0, 1\}$, $g:\{0, 1\}^k\times \{0,
1\}^k\to\{0, 1\}$. Write such an $F$ as $f\Box g$. An XOR function
is thus $f\Box\oplus$ with $k=1$. They showed that $Q(F)$ is
lower-bounded by the approximate polynomial degree of $f$ when
certain conditions on $k$ and $g$ are satisfied. Their bound  does
not apply to XOR functions that their $k$ should be sufficiently
large.

(2) Independent of \cite{yaoyun block composed functions},
Sherstov\cite{pattern matrix} studied what he called {\em pattern
matrices}. Those are block-composed functions for a fixed $g$, $g(x,
(y, w))=x_y\oplus w$, where $x\in\{0, 1\}^k$ is Alice's block, $(y,
w)$ is Bob's block with $y\in\{1, 2, ..., k\}$, $w\in\{0, 1\}$, and
$x_y$ is the $y$'th bit of $x$. Sherstov showed that for such
functions $Q(F)$ is lower bounded by the approximate polynomial
degree of $f$, multiplied by $\log k$. A XOR function is such a
function with $k=1$. However, Sherstov's lower bound vanishes on
this case.

(3) Aaronson~\cite{Aaronson} studied what he called {\em subset
problems}. Let $G$ be a group and $S$ be a subset of $G$. A subset
problem $\textrm{Subset}(G, S)$ is to decide if $x+y\in S$, where
$x, y\in G$ are the inputs of Alice and Bob, respectively. A XOR
function is a subset function with $G$ being the $n$-fold tensor
product of the 2 element finite field. Aaronson derived a general
lower bound on the {\em one-way} quantum communication complexity of
a subset problem. In contrast, we study the two-way communication
complexity.

We give the proofs for our main theorems in the next two sections before concluding
with a discussion on open problems.

\section{Deterministic communication complexity}

In this section we prove Theorem \ref{deterministic lower bounds}.
\begin{proofof}{Theorem~\ref{deterministic lower bounds}}
Let $H=[(-1)^{x\cdot y}]_{x, y\in\{0, 1\}^n}$ be the $2^n\times 2^n$
Hadamard Matrix and $D_F$ be the diagonal matrix with the diagonal
entries $[\tilde f(w)]_{w\in\{0, 1\}^n}$. Then $M_F=HD_FH$. Define
$\|\tilde f\|_0=\left|\{ \tilde f(w)\ne 0: w\in\{0, 1\}^n\}\right|$.
Since $H$ is orthogonal,
\begin{equation}\label{eqn:rank} \rank(M_F)=\|\tilde f\|_0.\end{equation}
By the symmetry assumption on $f$, $\tilde f$ is also symmetric.
That is, if $\tilde f(w)\ne 0$, $\tilde f(w')\ne 0$ for all $w'$
with $|w'|=|w|$. Therefore, by Lemma~\ref{lm:fourier}, to be proved
below, $\|\tilde f\|_0=2^{\Omega(n)}$. The theorem follows from
Eqn.~(\ref{eqn:rank}) and Lemma~\ref{lm:logrank}.
\end{proofof}

We now prove Lemma~\ref{lm:fourier}.

\begin{proofof}{Lemma~\ref{lm:fourier}}
Suppose that for a symmetric $f$, $\tilde f(w) = 0$ for all
$w\in\{0, 1\}^n$ with $|w|\in [n/8, 7n/8]$, we shall prove that $f$
is one of the four excluded functions.

For a polynomial $g$ and an integer $s$, let $T_s(g)$ denote the coefficient of
the monomial $x^s$ in $g$, and $G$ be the polynomial $\sum_{s=0}^n
f_sx^{n-s}$. Let
$f_k=f(1^k0^{n-k})$ and $\tilde{f}_k=\tilde{f}(1^k0^{n-k})$. Then
by the symmetry of $f$ and $\tilde f$,
  \[ \tilde{f}_k =\frac1{2^n}\sum_{y\in\{0, 1\}^n}\ f(y)(-1)^{1^k0^{n-k}\cdot y} =\frac1{2^n}\sum_y\ f(y)(-1)^{y_1+\cdots+y_k}.\]
Grouping $y$ by its Hamming weight, we have
  \[ \tilde{f}_k=\frac1{2^n}\sum_{s=0}^n\ f_s \sum_{|y|=s}(-1)^{y_1+\cdots+y_k}=\frac1{2^n}\sum_s f_s \sum_{t=0}^k (-1)^t{k\choose t}{n-k\choose s-t}.\]
Since $\sum_t (-1)^t{k\choose t}{n-k\choose s-t}$ is the coefficient
of the monomial $x^s$ in the polynomial $(1-x)^k(1+x)^{n-k}$ and
$f_s$ is that of $x^{n-s}$ in the polynomial $G$,
  \[ \tilde{f}_k=\frac1{2^n}\sum_s f_s T_s((1-x)^k(1+x)^{n-k})=\frac1{2^n}T_n(G\cdot(1-x)^k(1+x)^{n-k}).\]
Thus the assumption that  $\tilde{f}_k=0$ for all $t\leq k\leq n-t$
is equivalent to
$$T_n(G\cdot(1-x)^k(1+x)^{n-k})=0,\quad\textrm{for all $k$,
$t\leq k \leq n-t$}.$$
It follows that for
any $t\leq i, j\leq n-t$ with $i+j\leq n$,
\begin{eqnarray*}
  0&=&\sum_{s=i}^{n-j}T_n\left(G\cdot(1-x)^s(1+x)^{n-s}{n-i-j\choose
  s-i}\right)\\
  &=&T_n\left(G\cdot (1-x)^i(1+x)^j\sum_{s=i}^{n-j}{n-i-j\choose
  s-i}(1-x)^{s-i}(1+x)^{n-i-j-(s-i)}\right)\\
  &=&T_n(G\cdot(1-x)^i(1+x)^j\cdot 2^{n-i-j}).
\end{eqnarray*}
Therefore, for any $i, j$ with $t\leq i, j\leq n-t$ and $i+j\leq n$, we have
\begin{equation}\label{tnG}T_n(G\cdot(1-x)^i(1+x)^j)=0.\end{equation}
Let $u$ be an integer with $t\le u\le n/2$. We set $i=u$.
Setting $j=u$, and $j=u+1$, respectively, Eqn.~(\ref{tnG}) becomes $$T_n(G\cdot(1-x^2)^u) =
0,$$
and
$$0=T_n(G\cdot(1-x^2)^u(1+x))=T_n(G\cdot(1-x^2)^u)+T_{n-1}(G\cdot(1-x^2)^u).$$
Thus $$T_{n-1}(G\cdot(1-x^2)^u)=0.$$
Setting $j=u+2$ in Eqn.~(\ref{tnG}), we have
\begin{eqnarray*}
0&=&T_n(G\cdot(1-x^2)^u(1+x)^2)\\
&=&T_n(G\cdot(1-x^2)^u)+2T_{n-1}(G\cdot(1-x^2)^u)+T_{n-2}(G\cdot(1-x^2)^u).\end{eqnarray*}
Therefore
$$T_{n-2}(G\cdot(1-x^2)^u)=0.$$
Continuing this process till $i=u, j=n-u$, we have
\[ T_s(G\cdot(1-x^2)^u)=0,\quad\textrm{for all $s$, $2u\le s\le n$}.\]
 Expanding $G\cdot(1-x^2)^u$, we have
  \[T_s(G\cdot(1-x^2)^u)=T_s\left(\sum_kf_kx^{n-k}\sum_l{u\choose l}(-1)^lx^{2l}\right)
  =\sum_l{u\choose l}(-1)^lf_{n-s+2l}=0.\]
If $u$ is an odd prime, for all $l$, $1\leq l \leq u-1$,  $u\mid {u \choose l}$.
Thus
$$u\mid(f_{n-s}-f_{n-s+2u}),$$ since both $f_{n-s}$ and $f_{n-s+2u}$
are either $1$ or $0$. This implies that $f_{n-s}=f_{n-s+2u}$.
That is, for any odd prime $u\in[t, n/2]$, it holds that for any $s$ with $s\leq n-2u$,
$$f_s=f_{s+2u}.$$

Bertrand's Postulate\cite{prime} states that for any integer $m>3$,
there is at least one prime number between $m$ and $2m$. So we can
take two different primes $p,q\in [t, n/2]$ (recall that $t=n/8$) when $n\ge 32$, such that
$f_s=f_{s+2p}$, and $f_s=f_{s+2q}$. By Chinese Remainder Theorem, this
implies $f_s=f_{s+2}$ for all $s$. Then $f$ must be one of the four
functions excluded in the statement of the theorem.
\end{proofof}

\section{Randomized and quantum complexities}

In this section, we prove theorem \ref{quantum and randomized
complexity}. The proof has two parts, a lower bound proof and a
protocol. Both proofs are along the same line as those in
Huang et al.~\cite{hammingdistance} on the Hamming distance functions.
\begin{proposition}\label{quantum lower bounds}
  For any symmetric XOR function $F(x,y)=D(|x\oplus y|)$,
  $Q(F)=\Omega(r)$.
\end{proposition}

To prove this lower bound, we restrict the problem on those pairs of
inputs with an equal Hamming weight. For an integer $k$, where $0\leq
k\leq n$, define $X_k=Y_k=\{x\in\{0,1\}^n:|x|=k\}$. For a function
$S:\{0,1,\cdots, n\}\rightarrow \{0,1\}$, let $F_S$  be
the function $F_S(x,y)=S(|x\oplus y|)$. The restriction of $F_S$ on
$X_k\times Y_l$, where $0\le k, l\le n$, is denoted by $F_{k,l,S}$.
We shall use the following key lemma of Razborov~\cite{razborov symmetric predicates}.

\begin{lemma}[Razborov\cite{razborov symmetric predicates}]\label{Razborov disjointness lower bounds}
Suppose $k\leq n/4$ and $l\leq k/4$. Let $S:\{0,1,\cdots,
k\}\rightarrow \{0,1\}$ be any Boolean predicate such that
$S(l)\not=S(l-1)$. Let $f_{n,k,S}:X_k\times Y_k\rightarrow \{0,1\}$
be the function such that $f_{n,k,S}(x,y)=S(|x\wedge y|)$. Then
$Q(f_{n,k,S})=\Omega(\sqrt{kl})$.
\end{lemma}

\begin{proofof}{Proposition \ref{quantum lower bounds}}
For any XOR function $F(x,y)=S(|x\oplus y|)$, it can be decomposed
into two parts $F=F_{S_0}\wedge F_{S_1}$, where $F_{S_0}$ and
$F_{S_1}$ are XOR functions with the underlying functions $S_0,
S_1:\{0, 1\}^n\to\{0, 1\}$ defined as follows:
$S_0(x)=S(x)$ when $|x|$ is even, otherwise $S_0(x)=0$;
$S_1(x)=S(x)$ when $|x|$ is odd, otherwise $S_1(x)=0$. Since Alice
and Bob can compute the parity of $|x\oplus y|$ through a
$O(1)$-bits protocol,
\begin{equation}Q(F_{S_0}), Q(F_{S_1})\leq Q(F)+O(1)\label{D0D1D}.\end{equation}
Let $r_0^0=r_0(S_0)$ and $r_1^0=r_1(S_0)$. We have
$S_0(r_0^0-1)\not=S_0(r_0^0+1)$. We want to show
$Q(F_{S_0})=\Omega(r_0^0)$.

If $r_0^0\leq 3n/8$, this
will be proved by constructing another predicate $S':[k]\rightarrow
\{0,1\}$ for Lemma \ref{Razborov disjointness lower bounds} by
$S'(t)=S_0(2k-2t)$, here $k$ is a parameter determined later. We
define a predicate $f_{n,k,S'}$ on $X_k\times Y_k$ by
$f_{n,k,S'}=S'(x\wedge y)$. Because $|x\oplus y|=|x|+|y|-2|x\wedge
y|$ for any $x\in X_k$ and $y\in Y_k$, the two functions $F_{k,k,S}$ and
$f_{n,k,S'}$ are identical. Therefore,
$Q(f_{n,k,S'})=Q(F_{k,k,S_0})$.

Since $S_0(r_0^0-1)\not=S_0(r_0^0+1)$, then
$S'(k-(r_0^0-1)/2)\not=S'(k-(r_0^0+1)/2)$. For $r_0^0<3n/8$, let
$k=[2r_0^0/3]$, we have $k\leq n/4$ and $l\leq k/4$. By lemma
\ref{Razborov disjointness lower bounds}, we have $Q(F_{S_0})\geq
Q(F_{k,k,S_0})= Q(f_{n,k,S'}) \geq\Omega(\sqrt{kl})=\Omega(r_0^0)$.

When $r_0^0\geq 3n/8$, we will reduce to the previous case. Let
$n'=3n/4$ and consider the function $S_0':[n']\rightarrow\{0,1\}$
defined by $S_0'(x)=S_0(n-n'+x)$. Notice that the corresponding $r_0^0$
for $S_0'$ is $r_0^0-(n-n')$, which satisfies $r_0^0-(n-n')\leq
n/2-(n-n')=n/4\leq 9n/32=3n'/8$, and that $F_{S_0'}(x,y)=S_0'(|x\oplus
y|)$ is embedded to $F$. Therefore, $Q(F_{S_0})\geq
Q(F_{S_0'})=\Omega(r_0^0-(n-n'))=\Omega(n)=\Omega(r_0^0)$.

Consider the function $S_1'$ with $S_1'(x)=S_1(1+x)$.
Since $F_{S_1'}(x,y)=F_{S_1}(0x, 1y)$,
$F_{S_1'}$ is embedded in $F_{S_1}$. Let $r_0^1=r_0(S_1)$ and $r_1^1=r_1(S_1)$.
Then the corresponding $r_0(S_1')=r_0^1-1$. Similar to the case of
$S_0$, we have $Q(F_{S_1})\geq Q(F_{S_1'})\geq
\Omega(r_0^1-1)=\Omega(r_0^1)$.

Since $r_0=\max(r_0^0,r_0^1)$, Eqn.~[\ref{D0D1D}] implies that
$Q(F)=\Omega(r_0)$. Consider $\bar{S}(x)=S(n-x)$, then the
corresponding $r_0$ of $\bar S$ is exactly $r_1$. Since
$F_{\bar{S}}(x,y)=F(\bar{x}, y)$ (here $\bar{x}$ means the bit-wise
flipping of $x$), $F_{\bar{S}}$ and $F$ are actually equivalent so
that we have $Q(F)=Q(F_{\bar{S}})=\Omega(r_1)$. Combining the lower
bounds by $r_0$ and $r_1$, we have $Q(F)=\Omega(\max(r_0,r_1))$.
\end{proofof}

We now turn to the construction of a randomized protocol for
symmetric XOR functions.  Recall that the Hamming distance function
$\textrm{HAM}_{n,d}$ is defined as follows:
$\textrm{HAM}_{n,d}(x,y)=1 $ iff $|x\oplus y|> d$. Huang et
al.~\cite{hammingdistance} constructed an efficient randomized {\em
one-way} communication protocol for $\textrm{HAM}_{n, d}$, where Bob
is not allowed to send messages to Alice.
\begin{lemma}[Huang et al.
\cite{hammingdistance}]\label{protocolofhammingdistance}
There is a randomized one-way communication protocol for $\textrm{HAM}_{n, d}$ using $O(d\log d)$
bits.
\end{lemma}

We will make use of their protocol to prove the following.
\begin{proposition}
  There is a $O(r\log^2 r \log\log r)$ randomized protocol for any
  symmetric XOR function $F(x,y)=f(x\oplus y)$.
\end{proposition}

\begin{proof} We construct a public-coin randomized protocol as
following. By solving $\textrm{HAM}_{n,r}$ and $\textrm{HAM}_{n,n-r}$ using $O(r\log
r)$ bits (to make the final failure probability to be small, this
step will be repeated for constant times), Alice and Bob decide which of the
three intervals that $|x\oplus y|$ lies: $[r, n-r]$, $[0, r)$, or $[n-r,n]$, with high
probability. If $|x\oplus y|\in [r, n-r]$, by the definition of
$r$, $F$ only depends on the parity of $|x\oplus y|$, which can be computed in
$O(1)$ bits of communication.  If $|x\oplus y| \in [0,r)\cup (n-r,
n]$, Alice and Bob apply a binary search for $|x\oplus y|$. Each time
they check a Hamming distance instance $Ham_{n,k}$ for some $k\in
[0, r)\cap (n-r, n]$. The exact value of $|x\oplus y|$ can be determined in
$O(\log r)$ rounds. To output a correct answer with probability more
than $2/3$, it suffices to make sure that the failure probability is
$\le 1/(4\log r)$ in every round. This can be done by repeating
the Hamming distance instance $\Theta(\log \log r)$ times in each
round. By Lemma \ref{protocolofhammingdistance}, each round uses at
most $O(r\log r\log \log r)$ bits. The total cost of this protocol
is therefore $O(r\log^2r\log\log r)$.
\end{proof}

In the above protocol, Alice and Bob interactively send messages to
determine the exact $|x\oplus y|$ by binary search in $O(\log r)$
rounds. When Bob are not allowed to send information back to Alice,
they need to enumerate all possible $|x\oplus y|$ in the interval
$[0,r)\cap (n-r,n]$. Enumeration of $|x\oplus y|=d$ can be done by
solving two Hamming distance problems $\textrm{HAM}_{n,d-1}$ and $\textrm{HAM}_{n,d}$.
To obtain large success probability finally, each problem must
be repeated $O(\log r)$ times. This leads to the following.

\begin{proposition}
  There is a $O(r^2\log^2r)$ one-way randomized protocol for any
  symmetric XOR function $F(x,y)=f(x\oplus y)$.
\end{proposition}

The lower bound in Theorem \ref{quantum lower bounds} is still true
for one-way quantum communication because one-way complexity is
always larger than the corresponding two-way complexity. There remains
a quadratic gap between the lower bound and upper bound for the one-way
complexity.

\section{Discussion}

In addition to the above-mentioned question regarding one-way communication
complexity, we state two other open problems.
Our result implies that the correctness of the Log-Rank Conjecture
for the class of symmetric XOR functions. It will be interesting to extend this consequence
to the asymmetric case, and to make use of the fact
that $\rank(M_F)=\|\tilde f\|_0$ remains true for asymmetric $f$.

We may also consider the unbounded-error communication complexity of
 XOR functions. The unbound-error complexity, equivalent
with logarithm of sign-rank, has applications in other areas such as
circuit complexity, rigidity and PAC learning.
Sherstov\cite{sherstovunbounded} proved that the unbounded-error
complexity of $S(|x\wedge y|)$ is essentially
$|\{t:S(t)\not=S(t+1)\}|$. We conjecture that the unbounded-error complexity of
$S(|x\oplus y|)$ is essentially $|\{t:S(t)\not=S(t+2)\}|$. However,
Sherstov's approach does not seem to work for XOR functions because
the core technique used
--- pattern matrix cannot be embedded in a XOR function.

\end{document}